\documentclass[aip,pop,reprint,a4paper]{revtex4-1}
\usepackage{cancel}
\usepackage{amsmath}
\usepackage{amssymb}
\usepackage{amsthm}
\usepackage{graphicx}
\usepackage{dcolumn}
\usepackage{bm}
\usepackage{upgreek}
\usepackage{esint}

\usepackage{commands}

\theoremstyle{definition}
\newtheorem{definition}{Definition}

\theoremstyle{remark}
\newtheorem{remark}{Remark}

\theoremstyle{theorem}
\newtheorem{theorem}{Theorem}

\theoremstyle{lemma}
\newtheorem{lemma}{Lemma}

\theoremstyle{corollary}

\theoremstyle{proposition}
\newtheorem{proposition}{Proposition}

\draft 

\begin{document}


\title{Generalized Grad-Shafranov equation for non-axisymmetric MHD equilibria} 



\author{J. W. Burby}
\affiliation{Los Alamos National Laboratory, Los Alamos, New Mexico 87545, USA}

\author{N. Kallinikos}
\affiliation{Mathematics Institute, University of Warwick, Coventry CV4 7AL, U.K.}

\author{R. S. MacKay}
\affiliation{Mathematics Institute, University of Warwick, Coventry CV4 7AL, U.K.}



\date{\today}

\begin{abstract}
The structure of static MHD equilibria that admit continuous families of Euclidean symmetries is well understood. Such field configurations are governed by the classical Grad-Shafranov equation, which is a single elliptic PDE in two space dimensions. By revealing a hidden symmetry, we show that in fact all { smooth} solutions of the equilibrium equations { with non-vanishing pressure gradients away from the magnetic axis} satisfy a generalization of the Grad-Shafranov equation. In contrast to solutions of the classical Grad-Shafranov equation, solutions of the generalized equation are not automatically equilibria, but instead only satisfy force balance averaged over the one-parameter hidden symmetry. We then explain how the generalized Grad-Shafranov equation can be used to reformulate the problem of finding exact three-dimensional smooth solutions of the equilibrium equations as finding an optimal volume-preserving symmetry.
\end{abstract}


\maketitle 

\section{Introduction}
The ideal magnetohydrostatic equilibrium equations governing a magnetic field $\bm{B}$ supported by pressure $\mathsf{p}$,
\begin{gather}
(\nabla\times\bm{B})\times\bm{B} = \nabla\mathsf{p}\label{force_balance}\\
\nabla\cdot\bm{B}  =0,\label{divB}
\end{gather}
provide an approximate description of quiescent magnetized plasma on length scales comparable to the plasma's container. The main well-understood class of solutions that support non-zero pressure gradients $\nabla\mathsf{p}$ comprises those that admit rotational, translational, or screw symmetry. We will refer to such solutions as \emph{rigidly-symmetric}. Rigidly-symmetric solutions are sufficient to describe the nominal operating regimes for magnetic fusion experiments such as tokamaks, reverse-field pinches, and spheromaks. On the other hand, stellarators, which may offer the shortest path to a working fusion reactor,\cite{Boozer_stel_2020} are not rigidly symmetric. 
{ The fundamental design principle underlying stellarators is to purposefully break Euclidean symmetry in order to produce rotational transform without a significant plasma current. (See Ref.\,\onlinecite{Bhattacharjee_1992} for a clear exposition of the connection between magnetic axis geometry and rotational transform.) Therefore non-rigidly-symmetric solutions of Eqs.\,\eqref{force_balance}-\eqref{divB} are of considerable practical interest.}
{ Such} solutions have become loosely termed ``three-dimensional solutions'' due to the fact that rigidly-symmetric solutions are inherently two-dimensional. Indeed, assuming rigid symmetry, the PDE system\,\eqref{force_balance}-\eqref{divB} reduces to an elliptic PDE in two space dimensions for a single flux function $\psi$. This simplified PDE is known as the Grad-Shafranov (GS) equation; we refer the reader to Eq.\,(2.65) of Ref.\,\onlinecite{White_tc_2014} for the axisymmetric case, and Ref.\,\onlinecite{Smith_Rieman_2004} for the helical case.

Finding three-dimensional solutions with non-zero pressure gradients is inherently more complicated than finding solutions that admit continuous Euclidean symmetries. Grad's long-standing conjecture\cite{Grad_conj_1967} even postulates that continuous families of smooth three-dimensional solutions with pressure gradients cannot exist. If Grad's conjecture is correct, then it is still possible that three-dimensional solutions exist, but adiabatically-evolving three-dimensional solutions cannot. In the absence of adiabatically-evolving three-dimensional solutions, the theory of gradual profile evolution in stellarators would have to discard the constraints \,\eqref{force_balance}-\eqref{divB}, at least at their face value, in favor of a less-restrictive approximate notion of plasma equilibrium.

Although there are rigorous constructions\cite{Lortz_1970,Meyer_1958} of special smooth three-dimensional equilibria with closed $\bm{B}$-lines, { numerical studies\cite{Hudson_Kraus_2017,Mikhailov_2019,Kim_2020} demonstrating that singular current sheets can be avoided in three-dimensional equilibrium calculations, and numerically-verified\cite{Kim_2020}} asymptotic expansions for three-dimensional solutions with small shear,\cite{Weitzner_2016,Sengupta_Weitzner_2019} the deepest theory of three-dimensional solutions available today is due to Bruno and Laurence,\cite{Bruno_Laurence_1996} who established existence of nearly-axisymmetric three-dimensional solutions with stepped pressure profiles. Hudson and collaborators presently wield the Bruno-Laurence theory as a practical tool\cite{Hudson_spec_2012} for studying two- and three-dimensional\cite{Hudson_2011} equilibria in stellarators,\cite{Loizu_2015,Loizu_stel_2016} tokamaks, and reverse-field piches.\cite{Qu_2020} However, due to the discontinuities in the pressure function, the Bruno-Laurence solutions are not smooth. One can imagine that there might exist limits of Bruno-Laurence solutions with { continuous pressure and pressure gradients supported on a fat Cantor set of flux surfaces.\cite{Kraus_2017} (These are the pathological solutions envisaged by Grad.\cite{Grad_conj_1967})} While discontinuous, or at least nearly-discontinuous pressure profiles may be justified on physical grounds, such discontinuities introduce various complications in theories of nearly-quiescent plasma dynamics built upon leading-order magnetohydrostatic equilibrium. { (Even studies that employ specially-tuned continuous\cite{Hudson_Kraus_2017} or discontinuous\cite{Loizu_2015} rotational transform profiles to achieve continuous pressure profiles cannot ensure a high degree of smoothness for the magnetic field.)} For example, the magnetic moment adiabatic invariant for single-particle dynamics is conserved over larger time intervals in smoother magnetic fields.\cite{Kruskal_1962} Therefore the Bruno-Laurence theory does not eliminate the potential benefit of a satisfactory theory of smooth solutions, or smooth approximate solutions of Eqs.\,\eqref{force_balance}-\eqref{divB}.

The purpose of this Article is to establish { fundamental results concerning the structure of smooth three-dimensional equilibria and the} the existence of smooth approximate solutions of the magnetohydrostatic equations. Specifically we will show that a necessary condition for a smooth nondegenerate solution of Eqs.\,\eqref{force_balance}-\eqref{divB} to exist is that $\bm{B}$ admits a stream-function representation in terms of a single function $\psi$, and that $\psi$ satisfies a three-dimensional generalization of the Grad-Shafranov equation. Therefore smooth approximate solutions of the magnetohydrostatic equations may be constructed as solutions of the generalized Grad-Shafranov (GGS) equation. Here ``nondegenerate" refers to solutions with $(\bm{B},\nabla\times\bm{B})$ linearly-independent everywhere except on a single magnetic axis.

Where solutions of the classical GS equation are invariant under continuous families of Euclidean symmetries, solutions of the generalized GGS equation are invariant under continuous families of non-Euclidean volume-preserving symmetries. Otherwise, the GGS equation shares essentially all of its qualitative features, e.g.~ellipticity, variational principle, etc, with the classical Grad-Shafranov equation.  It is fair to say that all smooth nondegenerate solutions, three-dimensional or not, admit a ``hidden" symmetry. The symmetry is ``hidden" in the sense that an observer in the lab frame may see nothing obviously symmetrical about the field configuration. This suggests that, perhaps, the term ``three-dimensional solution" is actually a misnomer.

After establishing our main existence result, we will discuss the precise sense in which solutions of the GGS equation approximately satisfy Eqs.\,\eqref{force_balance}-\eqref{divB}. We will prove that while solutions of the GGS equation satisfy Eq.\,\eqref{divB} exactly, and possess nested toroidal flux surfaces, solutions of the GGS equation in general only satisfy the force-balance equation \eqref{force_balance} \emph{averaged over the symmetry} associated with the GGS equation. This observation will enable us to formulate a novel method for improving the accuracy of our approximate solutions, and therefore for searching for smooth three-dimensional MHD equilibria. This new approach reformulates the task of finding three-dimensional solutions in the space of pairs $(\bm{B},\mathsf{p})$ as a search for an optimal symmetry field, which is a divergence-free vector field with closed lines and constant period. 

We will express our results in terms of standard Heaviside vector calculus notation as much as possible. However, we take the liberty to use differential forms in our intermediate calculations whenever it is convenient. Readers who are unfamiliar with differential forms may consult the recent tutorial article \onlinecite{MacKay_tutorial_2020} for a primer, and Ref.\,\onlinecite{Abraham_2008} for an in-depth discussion.

\section{Definitions: averaged vector calculus}\label{def:circle_symmetry}
In this section, as well as the rest of the article, $Q\subset\mathbb{R}^3$ will stand for the ``plasma volume,'' i.e.~a region in $\mathbb{R}^3$ that is diffeomorphic to the solid torus $D^2\times S^1$. Here $D^2$ is the closed 2-dimensional unit disc and $S^1 = \mathbb{R}\text{ mod }2\pi$. The purpose of this section is to introduce several definitions that will be useful when discussing MHD equilibria in $Q$. 

We will begin by defining some concepts that are global in nature.
\begin{definition}
A \emph{circle-action} on $Q$ is a one-parameter family of { $C^\infty$} diffeomorphisms $\Phi_\theta:Q\rightarrow Q$ that satisfies the following properties for each $\theta,\theta_1,\theta_2\in S^1$:
\begin{align}
\Phi_{\theta+2\pi} &= \Phi_\theta\hspace{3em}(\text{periodicity})\\
\Phi_0 & = \text{id}_Q\hspace{2.7em}(\text{identity})\\
\Phi_{\theta_1+\theta_2}& = \Phi_{\theta_1}\circ \Phi_{\theta_2}\hspace{.2em}(\text{additivity}).
\end{align}
{ Moreover, the mapping $(\bm{x},\theta)\mapsto \Phi_\theta(\bm{x})$ is required to be $C^\infty$.} A \emph{volume-preserving circle-action} of $Q$ is a circle-action that does not change the volume of any subregion $U\subset Q$. In other words, $\Phi_\theta$ is a volume-preserving circle-action if
\begin{align}
\int_{U}\,d^3\bm{x} = \int_{\Phi_\theta(U)}\,d^3\bm{x}
\end{align}
for each $\theta\in S^1$ and each (relatively) open $U\subset Q$.
\end{definition}

\begin{definition}\label{def:infinitesimal_generator}
The \emph{infinitesimal generator} of a circle-action $\Phi_\theta$ is the vector field $\bm{u}$ on $Q$ defined by
\begin{align}
\bm{u}(\bm{x}) = \frac{d}{d\theta}\bigg|_0\Phi_\theta(\bm{x}).
\end{align} 
\end{definition}

\begin{remark}
The infinitesimal generator of a circle-action is always tangent to the boundary of $Q$. Moreover, a circle-action is volume-preserving if and only if its infinitesimal generator is divergence-free.
\end{remark}

\begin{definition}\label{def:averaged_metric}
Given a circle-action $\Phi_\theta$, the associated \emph{averaged metric} is the metric tensor $\overline{g}$ on $Q$ given by
\begin{align}\label{eq:averaged_metric}
\overline{g} = \fint \Phi_\theta^*g\,d\theta.
\end{align}
Here $g= g_{ij}\,dx^i\,dx^j$ is the usual flat metric on $\mathbb{R}^3$, $\fint = \frac{1}{2\pi}\int_0^{2\pi}$ is the normalized integral over $\theta$, and $\Phi_\theta^*g $ denotes pullback of $g$ along the mapping $\Phi_\theta$. 
\end{definition}
\begin{remark}
Set $\Phi_\theta = (\Phi_\theta^1,\Phi_\theta^2,\Phi_\theta^3)$ and let $[P(\theta)],[\overline{g}],[g]$ be the $3\times 3$ matrices with components
\begin{align}
[P(\theta)]_{ij} = & \partial_j\Phi_\theta^i\\
[\overline{g}]_{ij} = & \overline{g}_{ij}\\
[g]_{ij} = & g_{ij}.
\end{align}
Since the pullback $\Phi_\theta^*g = (g_{ij}\circ \Phi_\theta)\, \partial_k\Phi_\theta^i\,\partial_l\Phi_\theta^j \,dx^l\,dx^k$, these three matrices are related by the formula
\begin{align}
[\overline{g}] = \fint [P(\theta)]^T[g][P(\theta)]\,d\theta.
\end{align}
Aside from { indicating} how to compute $\overline{g}$, this formula also shows that $[\overline{g}]$ is positive definite. Therefore $\overline{g}$ is a genuine { $C^\infty$-smooth} metric tensor.
\end{remark}

\begin{definition}\label{def:volume_density}
Given a metric $h$ on $Q$, the \emph{$h$-volume} of a region $U\subset Q$ is given by
\begin{align}
\text{vol}_h(U) = \int_U \sqrt{\text{det}(h)}\,d^3\bm{x},
\end{align}
where $\text{det}(h)$ is the determinant of the matrix of $h$-coefficients in standard Euclidean coordinates. The positive scalar  $\rho_h\equiv\sqrt{\text{det}(h)}$ is the $h$-density.
\end{definition}
\begin{remark}
Note that the $h$-density is $1$ when { $h = g = dx^2 + dy^2 +dz^2$} is the standard flat metric on $\mathbb{R}^3$. 
\end{remark}

Next we will introduce some concepts that are local in nature.
\begin{definition}\label{def:averaged_dot_cross}
Given the averaged metric $\overline{g}$ associated with a circle-action $\Phi_\theta$, the \emph{averaged dot product} of a pair of vector fields $\bm{w}_1,\bm{w}_2$ is the scalar field on $Q$ given by
\begin{align}
\bm{w}_1\overline{\cdot}\,\bm{w}_2 = \overline{g}(\bm{w}_1,\bm{w}_2).
\end{align}
The \emph{averaged cross product} of $\bm{w}_1$ and $\bm{w}_2$ is the unique vector field $\bm{w}_1\overline{\times}\bm{w}_2$ that satisfies
\begin{align}
\bm{w}_1\overline{\times}\,\bm{w}_2\overline{\cdot}\,\bm{v} = \rho_{\overline{g}}\, \bm{w}_1\times\bm{w}_2\cdot\bm{v}, 
\end{align}
for each vector field $\bm{v}$ on $Q$. Recall that $\rho_{\overline{g}}$ is the $\overline{g}$-density introduced in Definition\,\ref{def:volume_density}.
\end{definition}
\begin{remark}
If $[\bm{w}_1],[\bm{w}_2],[\overline{g}]$ are the component matrices of $\bm{w}_1,\bm{w}_2,\overline{g}$, the averaged dot product may also be written
\begin{align}
\bm{w}_1\overline{\cdot}\,\bm{w}_2 = [\bm{w}_1]^T[\overline{g}][\bm{w}_2].
\end{align}
\end{remark}

\begin{definition}\label{def:averaged_grad_div_curl}
Given an averaged metric $\overline{g}$, the \emph{averaged gradient} of a { $C^\infty$} scalar field $\varphi$ is the unique vector field $\overline{\nabla}\varphi$ that satisfies
\begin{align}
\overline{\nabla}\varphi\,\overline{\cdot}\,\bm{v} = \nabla\varphi\cdot \bm{v}
\end{align}
for each vector field $\bm{v}$. The \emph{averaged divergence} of a { $C^\infty$} vector field $\bm{w}$ is the scalar field $\overline{\nabla}\,\overline{\cdot}\,\bm{w}$ given by
\begin{align}
\overline{\nabla}\,\overline{\cdot}\,\bm{w} = \rho_{\overline{g}}^{-1}\nabla\cdot(\rho_{\overline{g}}\,\bm{w}).
\end{align}
The \emph{averaged curl} of { the same} vector field $\bm{w}$ is the unique vector field $\overline{\nabla}\,\overline{\times}\bm{w}$ that satisfies
\begin{align}
\overline{\nabla}\,\overline{\times}\,\bm{w}\,\overline{\cdot}\,\bm{v}_1\overline{\times}\bm{v}_2 &= \bm{v}_1\,\overline{\cdot}\,\overline{\nabla}(\bm{w}\,\overline{\cdot}\,\bm{v}_2) - \bm{v}_2\,\overline{\cdot}\,\overline{\nabla}(\bm{w}\,\overline{\cdot}\,\bm{v}_1)\nonumber\\
 &- \bm{w}\,\overline{\cdot}\,[\bm{v}_1,\bm{v}_2],
\end{align}
for each pair of vector fields $\bm{v}_1,\bm{v}_2$. Here $[\bm{v}_1,\bm{v}_2] = \bm{v}_1\cdot\nabla\bm{v}_2 - \bm{v}_2\cdot\nabla\bm{v}_1$ is the usual commutator of vector fields. 
\end{definition}

The definitions of averaged vector calculus operators in this section should be compared with standard formulas for flat-metric vector calculus in curvilinear coordinates. 

\section{Supporting Lemmas}
This Section will establish some technical lemmas that will be useful when proving this Article's main result. Our presentation depends on concepts from exterior differential calculus, but we write the main conclusions in vector calculus language. For an exterior calculus tutorial for plasma physicists, see Ref.\,\onlinecite{MacKay_tutorial_2020}. 

\begin{lemma}[Structure Equations]\label{lemma:u_struct}
Let $\bm{u}$ be the infinitesimal generator of a volume-preserving circle-action $\Phi_\theta$. Set $R^2 = \bm{u}\,\overline{\cdot}\,\bm{u}$. Then we have the identities
\begin{gather}
\bm{u}\,\overline{\cdot}\,\overline{\nabla}R^2 = 0\\
\bm{u}\,\overline{\cdot}\,\overline{\nabla} \rho_{\overline{g}} = 0\\
\left(\overline{\nabla}\,\overline{\times}\,\bigg(\frac{\bm{u}}{R^2}\bigg)\right)\,\overline{\times}\,\bigg(\frac{\bm{u}}{R^2}\bigg) = 0.\label{force_free}
\end{gather}
\end{lemma}

\begin{proof}
The first identity follows immediately from
\begin{align}
\bm{u}\,\overline{\cdot}\,\overline{\nabla} R^2  &= L_{\bm{u}}R^2\nonumber\\
& = (L_{\bm{u}}\overline{g})(\bm{u},\bm{u}) + 2 \overline{g}(\bm{u},L_{\bm{u}}\bm{u}) = 0,
\end{align}
where $L_{\bm{u}}$ denotes the Lie derivative along $\bm{u}$, and we have used $L_{\bm{u}}\overline{\tau} = 0$ for any tensor $\overline{\tau} = \fint\Phi_\theta^*\tau\,d\theta$.
The second identity is similarly straightforward:
\begin{align}
0 = L_{\bm{u}} (\rho_{\overline{g}}\,d^3\bm{x}) = (L_{\bm{u}}\rho_{\overline{g}})\,d^3\bm{x},
\end{align}
since $\nabla\cdot\bm{u} = 0$.
To prove the third identity, first note that the $1$-form $\nu = \iota_{\bm{u}}\overline{g}$ satisfies $L_{\bm{u}} \nu  = \iota_{\bm{u}}L_{\bm{u}}\overline{g}  = 0$. Then recall the identity from Cartan calculus 
\begin{align}
L_{ f\bm{v}}\alpha = f L_{\bm{v}}\alpha + (\iota_{\bm{v}}\alpha)\,\mathbf{d}f,
\end{align}
where $\mathbf{d}$ is the exterior derivative, $\bm{v}$ is any vector field, $\iota_{\bm{v}}$ is the interior product (contraction with the first index), and $f$ is any $0$-form. Finally observe that on the one hand
\begin{align}
L_{\bm{u}/R^2}( \iota_{\bm{u}/R^2}\overline{g})& = (1/R^2) L_{\bm{u}} ( \iota_{\bm{u}/R^2}\overline{g}) + \mathbf{d}(1/R^2) \nonumber\\
&= \mathbf{d}(1/R^2),
\end{align}
by the preceding remarks, while on the other hand
\begin{align}
L_{\bm{u}/R^2}( \iota_{\bm{u}/R^2}\overline{g}) =& \mathbf{d}(1/R^2) + \iota_{\bm{u}/R^2}\mathbf{d}( \iota_{\bm{u}/R^2}\overline{g})\nonumber\\
=& \mathbf{d}(1/R^2) + \iota_{\bm{u}/R^2}\iota_{\overline{\nabla}\,\overline{\times}(\bm{u}/R^2)}\,\rho_{\overline{g}}d^3\bm{x},
\end{align}
by Cartan's formula for the Lie derivative. (Note that, while $\overline{g}$ is not a differential form, $\iota_{\bm{u}/R^2}\overline{g}$ is.) It follows that
\begin{align}
0 = &\iota_{\bm{v}}\iota_{\bm{u}/R^2}\iota_{\overline{\nabla}\,\overline{\times}(\bm{u}/R^2)}\,\rho_{\overline{g}}d^3\bm{x}\nonumber\\
 = & \rho_{\overline{g}}\left(\overline{\nabla}\,\overline{\times}\bigg(\frac{\bm{u}}{R^2}\bigg)\right)\times \bigg(\frac{\bm{u}}{R^2}\bigg)\cdot \bm{v}\nonumber\\
 = & \left(\overline{\nabla}\,\overline{\times}\bigg(\frac{\bm{u}}{R^2}\bigg)\right)\,\overline{\times} \bigg(\frac{\bm{u}}{R^2}\bigg)\,\overline{\cdot}\,\bm{v},
\end{align}
for each vector field $\bm{v}$, which is equivalent to the desired identity due to the non-degeneracy of $\overline{g}$.
\end{proof}

\begin{remark}
We use the notation $R^2$ for $\bm{u}\overline{\cdot}\bm{u}$ to recall the case of axisymmetry, where the infinitesimal generator of azimuthal rotations is $\bm{u} = R\,\bm{e}_\phi$, with $R$ the cylindrical radius and $\bm{e}_\phi$ the azimuthal unit vector. 
\end{remark}

\begin{lemma}[stream function representation for symmetric magnetic fields]\label{lemma:stream_function}
Suppose $\bm{B}$ is a { $C^\infty$}  divergence-free vector field on $Q$ that satisfies $[\bm{u},\bm{B}] = 0$, where $\bm{u}$ is the infinitesimal generator of a volume-preserving circle-action $\Phi_\theta$. Then there exists a { $C^\infty$}  function $\psi: Q\rightarrow\mathbb{R}$, called \emph{the stream function}, such that 
\begin{align}
\bm{u}\times\bm{B} = \nabla\psi,
\end{align}
or, equivalently,
\begin{align}
\bm{u}\,\overline{\times}\,\bm{B} = \rho_{\overline{g}}\,\overline{\nabla}\psi.
\end{align}
\end{lemma}
\begin{proof}
Set $\beta = \iota_{\bm{B}}d^3\bm{x}$. Because $\bm{B}$ is divergence-free, the $2$-form $\beta$ is closed. In fact, because the second de Rham cohomology group of the solid torus is trivial, there must be a $1$-form $\alpha$ such that $\beta = \mathbf{d}\alpha$. (Note that any physical $\bm{B}$ must admit a vector potential in any region devoid of magnetic monopoles. However, on purely mathematical grounds $\nabla\cdot\bm{B} = 0$ in a region with non-trivial second cohomology does not always imply $\bm{B} = \nabla\times\bm{A}$. Therefore it would be more physical to replace the equilibrium condition $\nabla\cdot \bm{B} = 0$ with $\bm{B} = \nabla\times\bm{A}$, but we digress.) In light of the fact that $\Phi_\theta$ is volume-preserving and $\bm{B}$-preserving (because $[\bm{u},\bm{B}] =L_{\bm{u}}\bm{B} = 0$), $\beta$ therefore satisfies
\begin{align}
\beta = \fint \Phi_\theta^*\beta\,d\theta = \mathbf{d}\fint \Phi_\theta^*\alpha\,d\theta = \mathbf{d}\overline{\alpha},\label{eq:av_b}
\end{align}
where $\overline{\alpha}= \fint \Phi_\theta^*\alpha\,d\theta$ is the $S^1$-average of $\alpha$. 

Set $\psi = \iota_{\bm{u}}\overline{\alpha}$. By Eq.\,\eqref{eq:av_b}, the exterior derivative of $\psi$ is given by
\begin{align}
\mathbf{d}\psi= L_{\bm{u}} \overline{\alpha} -\iota_{\bm{u}}\mathbf{d}\overline{\alpha} =- \iota_{\bm{u}}\beta.
\end{align}
By converting this expression into vector calculus notation with respect to either the usual flat metric $g$ or the averaged metric $\overline{g}$, we obtain the desired formulas.

\end{proof}

\section{The main result}
In this Section we will establish our main result, namely that all nondegenerate solutions of the ideal equilibrium equations must satisfy a generalization of the Grad-Shafranov equation. To prove this result, we will first use Hamada coordinates to demonstrate that all solutions possess a hidden volume-preserving symmetry. We will also provide a coordinate-free proof.

\begin{definition}
{ A pair $(\bm{B},\mathsf{p})$, where $\bm{B}$ is a $C^\infty$ vector field on $Q$ and $\mathsf{p}$ is a $C^\infty$ function on $Q$, is a \emph{nondegenerate solution} of the MHD equilibrium equations on $Q$ if
\begin{gather}
(\nabla\times\bm{B})\times\bm{B} = \nabla\mathsf{p}\\
\nabla\cdot\bm{B} = 0,
\end{gather}
the normal component $\bm{B}\cdot\bm{n}=0$ on $\partial Q$, the zero locus $\ell$ of $\nabla\mathsf{p}$ is diffeomorphic to $S^1$, and the Hessian $\nabla^2\mathsf{p}$ along $\ell$ is sign-definite normal to $\ell$. }
\end{definition}
\begin{remark}
{ The requirement that $\nabla^2\mathsf{p}$ be sign-definite normal to $\ell$ is equivalent to the following condition. If $\bm{w}$ is any nowhere-vanishing vector field that is never tangent to $\ell$ then $\bm{w}\cdot \nabla^2\mathsf{p}\cdot \bm{w}$ is nowhere-vanishing along $\ell$.}
\end{remark}

\begin{theorem}\label{hidden_symmetry_thm}
Given any nondegenerate solution $(\bm{B},\mathsf{p})$ of the MHD equilibrium equations on $Q$, there is a volume-preserving circle-action $\Phi_\theta$ with { nowhere-vanishing} infinitesimal generator $\bm{u}$ that satisfies $[\bm{u},\bm{B}] =[\bm{u},\bm{J}]= 0$,  where $\bm{J} = \nabla\times\bm{B}$.
\end{theorem}

\begin{proof}[proof using Hamada coordinates]
Because the equilibrium is non-degenerate the level sets of $\mathsf{p}$ foliate $Q$ by toroidal surfaces. Use $V$ to denote flux-surface volume. Let $\iota(V)$ and $\eta(V)$ denote the rotational transforms for $\bm{B}$ and $\bm{J} $, respectively. In Hamada coordinates\cite{Hamada_1962} $(V,\zeta_1,\zeta_2)$, $\bm{B}$ and $\bm{J}$ have the form
\begin{align}
\bm{B} =& B^1(V)\,\partial_{\zeta_1} + B^2(V)\,\partial_{\zeta_2}\\
\bm{J} = & J^1(V)\,\partial_{\zeta_1} + J^2(V)\,\partial_{\zeta_2},
\end{align}
where $B^2/B^1 = \iota(V)$, $J^2/J^1 = \eta(V)$. Consider a linear combination of $\bm{B}$ and $\bm{J}$ of the form
\begin{align}
\bm{u} = & c_1(V)\,\bm{B} + c_2(V)\bm{J}\nonumber\\
& = (c_1 B^1 + c_2 J^1)\,\partial_{\zeta_1} + (c_1 B^2 + c_2 J^2)\,\partial_{\zeta_2}.
\end{align}
If $(c_1,c_2)$ is defined to be the unique solution of
\begin{align}
c_1 B^1 + c_2 J^1 =& n\\
c_1 B^2 + c_2 J^2 = & m,
\end{align}
where $n,m\in\mathbb{Z}$, then $\bm{u}$ is the infinitesimal generator of a circle-action. Because $c_1,c_2$ are flux functions, $\nabla\cdot\bm{u} = 0$, which implies that the circle-action is volume preserving. Moreover,
\begin{align}
L_{\bm{u}}\bm{B} = -L_{\bm{B}}(c_1\bm{B} + c_2\bm{J}) = 0,
\end{align}
which shows $[\bm{u},\bm{B}] =0$. Computing $L_{\bm{u}}\bm{J}$ in the same manner shows $[\bm{u},\bm{J}] = 0.$

{ This argument is enough to establish the existence of the desired circle action away from the zero locus of $\nabla\mathsf{p}$. The crucial issue is that Hamada coordinates necessarily exhibit a coordinate singularity along that set. We discuss this issue further at the end of the following coordinate-free alternative proof. }

\end{proof}

\begin{proof}[coordinate-free proof]
Let $\mathsf{p}_{\text{pb}}$ denote the (constant) value of pressure on $\partial Q$ and $\mathsf{p}_{\text{axis}}$ the on-axis pressure. { (Here, ``on-axis" means on the zero locus of $\nabla\mathsf{p}$.)} For concreteness, assume $\mathsf{p}_{\text{pb}} < \mathsf{p}_{\text{axis}}$. (The argument is essentially the same with the opposite ordering.) Fix an arbitrary $p\in(\mathsf{p}_{\text{pb}}, \mathsf{p}_{\text{axis}})$. The level set $S_p = \{\bm{x}\in Q\mid \mathsf{p}(\bm{x}) = p\}$ must be a submanifold whose connected components are diffeomorphic to $2$-tori. In fact $S_p$ must be connected, for if it had two components $S_1,S_2$ then, without loss of generality, we may assume $S_2$ is contained in the volume enclosed by $S_1$, which would imply that there is a $\nabla \mathsf{p}$-line that intersects $S_p$ in two distinct points. It follows that the $S_p$ are nested toroidal surfaces that degenerate to a circle when $p = \mathsf{p}_{\text{axis}}$. We may therefore introduce smooth families of parameterized closed curves $\gamma^T_p,\gamma^P_p$ such that $\gamma^T_p,\gamma^P_p$ are generators for { the fundamental group} $\pi_1(S_p)$ when $p\in[\mathsf{p}_{\text{pb}}, \mathsf{p}_{\text{axis}})$, and when $p = \mathsf{p}_{\text{axis}}$ the curve $\gamma^T_p$ generates $\pi_1(S_{\mathsf{p}_{\text{axis}}})$ and $\gamma^P_p$ is constant.

Fix $p\in(\mathsf{p}_{\text{pb}}, \mathsf{p}_{\text{axis}})$, as before. The volume form $\Omega$ on $Q$ induces a $2$-form $\mu_p$ on $S_p$ such that if $\mu$ is any $2$-form on $Q$ with $\Omega = dV\wedge \mu/(2\pi)^2$, with $V$ the pressure surface volume, then $\iota_{S_p}^*\mu =\mu_p$.
Because $\bm{J}$ and $\bm{B}$ are divergence-free, $L_{\bm{B}_p}\mu_p = L_{\bm{J}_p}\mu_p =0$, where $\bm{B}_p,\bm{J}_p$ are $\bm{B}$ and $\bm{J}$ restricted to $S_p$. Therefore the $1$-forms $\iota_{\bm{B}_p}\mu_p$, $\iota_{\bm{J}_p}\mu_p$ are closed and determine De Rham cohomology classes $[\iota_{\bm{B}_p}\mu_p], [\iota_{\bm{J}_p}\mu_p]\in H^2_{\text{dR}}(S_p)$. Because the curves $\gamma^T_p,\gamma^P_p$ induce a natural basis for $H^2_{\text{dR}}(S_p)\approx \mathbb{R}^2$, these cohomology classes may be regarded as tuples $[\iota_{\bm{B}_p}\mu_p] = (b^T_p,b^P_p) = b(p)$ and $[\iota_{\bm{J}_p}\mu_p] = (j^T_p,j^P_p) = j(p)$. Because $\bm{B}$ and $\bm{J}$ are linearly-independent and simultaneously conjugate to linear flows { by the Liouville-Arnold theorem}, the classes $b(p)$ and $j(p)$ must be linearly independent as well. 

Let $c_b,c_j$ be smooth, as-yet undetermined functions of pressure, and define the divergence-free field $\bm{u} = c_b\,\bm{B} + c_j\,\bm{J}$. By linearity, the cohomology class $[\iota_{\bm{u}_p}\mu_p] = (U^T_p,U^P_p) = U(p)$ is given by $U(p) = c_b(p)\,b(p)+c_j(p)\,j(b)$. If $U(p) = (n,m)$, for integers $(n,m)$ then $\bm{u}$ will generate a volume-preserving circle-action that commutes with both $\bm{B}$ and $\bm{J}$. Because $(b(p),j(p))$ is a basis for $\mathbb{R}^2$, we can choose $c_b,c_j$ to be the coefficients of $(n,m)$ with respect to that basis for any $(n,m)$. Therefore we obtain a distinct circle-action of the desired type for each $(n,m)\in\mathbb{Z}^2$. { Note that because the construction of Hamada coordinates may be seen as the same application of the Liouville-Arnold theorem as above, the circle actions constructed in this alternative proof are the same as those found earlier using Hamada coordinates.}

{ The preceding argument has two shortcomings (1) it does not address the regularity of the circle action on the magnetic axis, and (2) it does not show that the generator $\bm{u}$ is nowhere vanishing on the magnetic axis. However, a forthcoming publication\cite{BMD_2020} proves that, under the hypotheses of this Theorem, the circle action is in fact $C^\infty$ everywhere for any $(n,m)$ and that $(n,m)$ may always be chosen to avoid a purely-poloidal rotation. In particular, in a neighborhood of the magnetic axis there exist smooth coordinates $(x,y,\phi)\in D^2\times S^1$ such that rotations in the $(x,y)$-plane and translations along $\phi$ are volume-preserving symmetries for $\bm{B}$ and $\bm{J}$.  We remark that on-axis regularity of three-dimensional MHD equilibria appears to be a delicate subject. Weitzner has shown in Ref.\,\onlinecite{Weitzner_2016} that while near-axis expansions of MHD equilibria may be carried out to all orders, such expansions can generate weak singularities on the magnetic axis. It is currently unclear whether such singularities can be avoided to all orders.}



\end{proof}

\begin{remark}
{ 
The preceding proof actually shows a stronger result. Namely, there exists a family of volume-preserving circle actions parameterized by pairs of integers $(n,m)\neq (0,0)$ with infinitesimal generators that are non-zero almost everywhere and that commute with $\bm{B}$ and $\bm{J}$. The infinitesimal generator corresponding to a purely poloidal rotation would necessarily be zero along the magnetic axis. A consequence of Lemma \ref{lemma:stream_function} is that the stream function associated with each of these alternative circle actions will be an integer combination of the toroidal and poloidal flux functions.}
\end{remark}

Presence of hidden symmetry is significant because of the following result, which generalizes the usual derivation of the Grad-Shafranov equation. 

\begin{theorem}[generalized Grad-Shafranov equation]\label{theorem:GS_equation}
Let $(\bm{B},\mathsf{p})$ be any nondegenerate solution of the MHD equilibrium equations on $Q$. If $\Phi_\theta$ is a volume-preserving circle-action whose infinitesimal generator $\bm{u}$ { is nowhere vanishing and} satisfies $[\bm{u},\bm{B}] = [\bm{u},\bm{J}] = 0$ then $\bm{B}$ may be written
\begin{align}
\bm{B} = C(\psi)\,\frac{\bm{u}}{R^2} + \rho_{\overline{g}}\frac{\overline{\nabla}\psi\,\overline{\times}\,\bm{u}}{R^2},\label{stream_rep}
\end{align}
where $C:\mathbb{R}\rightarrow\mathbb{R}$ is a smooth function of a single variable that satisfies $C(\psi) = \bm{u}\overline{\cdot}\bm{B}$, and $\psi:Q\rightarrow\mathbb{R}$ is a $C^\infty$ function that satisfies \emph{the generalized Grad-Shafranov equation},
\begin{align}
&-\overline{\nabla}\,\overline{\cdot}\,(R^{-2}\rho_{\overline{g}}\,\overline{\nabla}\psi) +C(\psi)\,\,\frac{\bm{u}}{R^2}\,\overline{\cdot}\,\overline{\nabla}\,\overline{\times}\,\bigg(\frac{\bm{u}}{R^2}\bigg)\nonumber\\
&=\frac{ p^\prime(\psi)+R^{-2}\,C(\psi)\,C^\prime(\psi)}{\rho_{\overline{g}}},\label{ggs_thm2}
\end{align}
with $\psi = \text{const.}$ on $\partial Q$.
Recall that $R^2 = \bm{u}\,\overline{\cdot}\,\bm{u}$, $\rho_{\overline{g}}$ is the $\overline{g}$-density, and $\bm{J} = \nabla\times\bm{B}$.
\end{theorem}

\begin{proof}
Because $\bm{B}$ is a divergence-free vector field on $Q$ satisfying $[\bm{u},\bm{B}]$, with $\bm{u}$ the infinitesimal generator of a volume-preserving circle-action, Lemma \ref{lemma:stream_function} implies that there is a stream function $\psi$ such that
\begin{align}
\bm{u}\,\overline{\times}\,\bm{B} = \rho_{\overline{g}}\,\overline{\nabla}\psi. \label{fourfour}
\end{align}
Therefore
\begin{align}
(\bm{u}\,\overline{\times}\,\bm{B})\,\overline{\times}\,\bm{u} &= R^2\,\bm{B} - (\bm{u}\,\overline{\cdot}\,\bm{B})\,\bm{u} = \rho_{\overline{g}}\,\overline{\nabla}\psi\,\overline{\times}\bm{u}\nonumber\\
\Rightarrow  \bm{B} &= (\bm{u}\overline{\cdot}\bm{B})\,\frac{\bm{u}}{R^2} + \rho_{\overline{g}}\frac{\overline{\nabla}\psi\,\overline{\times}\bm{u}}{R^2}.\label{thm_sr}
\end{align}
The proof will therefore be complete if we establish (A) $\bm{u}\overline{\cdot}\bm{B}=C(\psi)$ for some $C$, and (B) that $\psi$ satisfies the generalized Grad-Shafranov equation.
\\ \\
\noindent (A): Because the equilibrium is nondegenerate, we may prove that $\bm{u}\overline{\cdot}\bm{B}=C(\psi)$ by showing $L_{\bm{B}} C = L_{\bm{J}}C = 0$. By the Leibniz property for the Lie derivative,
\begin{align}
L_{\bm{B}}(\bm{u}\overline{\cdot}\bm{B})&= L_{\bm{B}} ([\iota_{\bm{B}}\overline{g}](\bm{u}))\nonumber\\
& =[L_{\bm{B}}(\iota_{\bm{B}}\overline{g})](\bm{u}) + [\iota_{\bm{B}}\overline{g}](L_{\bm{B}}\bm{u})\nonumber\\
& = [L_{\bm{B}}(\iota_{\bm{B}}\overline{g})](\bm{u}) \\
L_{\bm{J}}(\bm{u}\overline{\cdot}\bm{B}) & = L_{\bm{J}}([\iota_{\bm{B}}\overline{g}](\bm{u}))\nonumber\\
& = [L_{\bm{J}}(\iota_{\bm{B}}\overline{g})](\bm{u}) + [\iota_{\bm{B}}\overline{g}](L_{\bm{J}}\bm{u}) \nonumber\\
&= [L_{\bm{J}}(\iota_{\bm{B}}\overline{g})](\bm{u}).
\end{align}
Therefore $L_{\bm{B}}C = L_{\bm{J}}C=0$ if and only if $ \gamma = [L_{\bm{B}}(\iota_{\bm{B}}\overline{g})](\bm{u})=0$ and $\eta =  [L_{\bm{J}}(\iota_{\bm{B}}\overline{g})](\bm{u}) = 0$. In order to show that $\gamma$ and $\eta$ are in fact zero, we first write the MHD equilibrium equation in terms of differential forms as
\begin{align} 
\iota_{\bm{B}}\mathbf{d}(\iota_{\bm{B}}g) = \mathbf{d}\mathsf{p}.
\end{align}
Next we apply the pullback $\Phi_\theta^*$ to both sides of the equation and average over the parameter $\theta$. The result is
\begin{align}
\iota_{\bm{B}}\mathbf{d}(\iota_{\bm{B}}\overline{g}) = \mathbf{d}\mathsf{p}.\label{eq:averaged_gs}
\end{align}
Now notice that, by Cartan's formula,
\begin{align}
L_{\bm{B}}(\iota_{\bm{B}}\overline{g}) &= \mathbf{d}(\mathsf{p}+\overline{g}(\bm{B},\bm{B}))\\
L_{\bm{J}}(\iota_{\bm{B}}\overline{g})& = \iota_{\bm{J}}\mathbf{d}\iota_{\bm{B}}\overline{g} + \mathbf{d}(\overline{g}(\bm{J},\bm{B})) = \mathbf{d}(\overline{g}(\bm{J},\bm{B})) ,
\end{align}
where we have used $\iota_{\bm{J}}\Omega = \mathbf{d}\iota_{\bm{J}}g = \mathbf{d}\iota_{\bm{J}}\overline{g}$, with $\Omega$ the standard Euclidean volume form on $Q$. It follows that $\gamma$ and $\eta$ are given by
\begin{align}
\gamma & = L_{\bm{u}}(\mathsf{p}+\overline{g}(\bm{B},\bm{B})) = 0\\
\eta & = L_{\bm{u}}(\overline{g}(\bm{J},\bm{B})) = 0,
\end{align}
which implies $L_{\bm{B}}C = L_{\bm{J}}C = 0$, as desired.
\\ \\
\noindent (B): To establish the generalized Grad-Shafranov equation, we begin with the following three flux relations,
\begin{align}
\bm{u}\overline{\times}\bm{J} &= \rho_{\overline{g}}\overline{\nabla}C\label{Cflux}\\
\bm{u}\overline{\times}\bm{B}& = \rho_{\overline{g}}\overline{\nabla}\psi\label{psiflux}\\
\bm{J}\overline{\times}\bm{B} &= \rho_{\overline{g}}\overline{\nabla}\mathsf{p}.\label{pflux}
\end{align}
Equation \eqref{Cflux} follows from $\iota_{\bm{u}}\iota_{\bm{J}}\Omega = \iota_{\bm{u}}\mathbf{d}\iota_{\bm{B}}\overline{g} = -\mathbf{d}(\bm{u}\overline{\cdot}\bm{B})$; Eq.\,\eqref{psiflux} is just Eq.\,\eqref{fourfour}; while Eq.\,\eqref{pflux} is just the averaged force balance equation \eqref{eq:averaged_gs} expressed in averaged vector-calculus notation.
Substituting Eq.\,\eqref{thm_sr} into Eq.\,\eqref{pflux} and using Eq.\,\eqref{Cflux} implies $\bm{J}\overline{\cdot}(\bm{u}/R^2) = C\,C^\prime/R^2 + p^\prime.$ Taking the averaged divergence of $R^{-2}$ times Eq.\,\eqref{psiflux} then gives 
\begin{align}
\overline{\nabla}\overline{\cdot}\left(R^{-2}\rho_{\overline{g}}\overline{\nabla}\psi\right) &= \bm{B}\overline{\cdot}\overline{\nabla}\overline{\times}\left(\frac{\bm{u}}{R^2}\right) -( \overline{\nabla}\overline{\times}\bm{B})\overline{\cdot}\frac{\bm{u}}{R^2}\nonumber\\
& = C \frac{\bm{u}}{R^2}\overline{\cdot}\overline{\nabla}\overline{\times}\left(\frac{\bm{u}}{R^2}\right) - \frac{\bm{J}\overline{\cdot}\bm{u}}{\rho_{\overline{g}}R^2},
\end{align}
which is equivalent to the GGS equation \eqref{ggs_thm2}.

\end{proof}

\begin{remark}
{ According to Eq.\,\eqref{stream_rep}, the magnetic field is determined entirely by the circle action $\Phi_\theta$, the single-variable function $C(\psi)$, and the function $\psi$. Therefore all properties of $\bm{B}$ may be determined from that data. In particular the rotational transform $\iota(\psi)$ may be determined using the formula
\begin{align}
\iota(\psi) = -\frac{\oint_{\gamma_T} \frac{u}{R^2}\,\overline{\cdot}\,d\bm{x} + \oint_{\gamma_T}\frac{C}{\rho_{\overline{g}}\,\overline{\nabla}\psi\,\overline{\cdot}\,\overline{\nabla}\psi}\,\frac{\bm{u}}{R^2}\,\overline{\times}\,\overline{\nabla}\psi\,\overline{\cdot}\,d\bm{x}
 }{\oint_{\gamma_P} \frac{u}{R^2}\,\overline{\cdot}\,d\bm{x} + \oint_{\gamma_P} \frac{C}{\rho_{\overline{g}}\,\overline{\nabla}\psi\,\overline{\cdot}\,\overline{\nabla}\psi}\,\frac{\bm{u}}{R^2}\,\overline{\times}\,\overline{\nabla}\psi\,\overline{\cdot}\,d\bm{x}
},
\end{align}
where $\gamma_T$ and $\gamma_P$ are any toroidal and poloidal loops contained in the surface $\psi(\bm{x}) = \psi$. (The particular choice of toroidal and poloidal loop is immaterial because the $1$-form $\vartheta = \frac{u}{R^2}\,\overline{\cdot}\,d\bm{x}+\frac{C}{\rho_{\overline{g}}\,\overline{\nabla}\psi\,\overline{\cdot}\,\overline{\nabla}\psi}\,\frac{\bm{u}}{R^2}\,\overline{\times}\,\overline{\nabla}\psi\,\overline{\cdot}\,d\bm{x}$ is closed when pulled back to a $\psi$-surface.) We may also obtain the parallel current in the averaged metric according to $\bm{J}\,\overline{\cdot}\, \bm{B} = p^\prime\,C + C^\prime\,|\bm{B}|^2$.
}
\end{remark}

In light of Theorems \ref{hidden_symmetry_thm} and \ref{theorem:GS_equation}, we may now state the following remarkable fact about arbitrary smooth solutions of the ideal MHD equilibrium equations, which is the main result of this Article.

\begin{theorem}
Let $(\bm{B},\mathsf{p})$ be any nondegenerate solution of the MHD equilibrium equations on $Q$. There exists a volume-preserving circle-action with infinitesimal generator $\bm{u}$, a $C^\infty$ function $\psi:Q\rightarrow\mathbb{R}$, and a smooth single-variable function $C:\mathbb{R}\rightarrow \mathbb{R}$ such that the magnetic field $\bm{B}$ may be written 
\begin{align}
\bm{B} = C(\psi)\,\frac{\bm{u}}{R^2} + \rho_{\overline{g}}\frac{\overline{\nabla}\psi\,\overline{\times}\,\bm{u}}{R^2},\label{stream_rep_final}
\end{align}
$C(\psi) = \bm{u}\,\overline{\cdot}\,\bm{B}$, and $\psi$ solves the generalized Grad-Shafranov (GGS) equation,
\begin{align}
&-\overline{\nabla}\,\overline{\cdot}\,(R^{-2}\rho_{\overline{g}}\,\overline{\nabla}\psi) +C(\psi)\,\,\frac{\bm{u}}{R^2}\,\overline{\cdot}\,\overline{\nabla}\,\overline{\times}\,\bigg(\frac{\bm{u}}{R^2}\bigg)\nonumber\\
&=\frac{ p^\prime(\psi)+R^{-2}\,C(\psi)\,C^\prime(\psi)}{\rho_{\overline{g}}},\label{ggs_eqn}
\end{align}
with $\psi = \text{const.}$ on $\partial Q$.
\end{theorem}

\begin{remark}
{ Because $C(\psi) = \bm{u}\,\overline{\cdot}\,\bm{B}$, the function $C$ has the following physical interpretation. Let $\gamma(\lambda)$ be an integral curve for the infinitesimal generator $\bm{u}$, i.e. a parameterized curve that satisfies $\partial_\lambda\gamma = \bm{u}(\gamma(\lambda))$. Let $S_\gamma$ be a ribbon-like surface spanning from the magnetic axis $\ell$ to the integral curve $\gamma$. The current flowing through $S_\gamma$ is, by Stokes' theorem, $I = \int_{S_\gamma}\nabla\times\bm{B}\cdot d\bm{S} = \oint_\gamma \bm{B}\cdot d\bm{x} - \oint_\ell \bm{B}\cdot d\bm{x}$. Because $\gamma$ must be contained in some constant-$\psi$ surface, and the curve $\gamma$ may be interpreted as a toroidal circuit, $I$ may be interpreted as the total poloidal current flowing in the volume bounded by the $\psi$-surface and the zero-locus of $\nabla\mathsf{p}$. But by the change-of-variables formula for line integrals $\oint_\gamma \bm{B}\cdot d\bm{x} = \oint_\gamma \bm{B}\,\overline{\cdot}\,d\bm{x} = \int_0^{2\pi} \bm{B}(\gamma(\lambda))\,\overline{\cdot}\,\bm{u}(\gamma(\lambda))\,d\lambda = 2\pi\, C(\psi(\gamma))$. Therefore the total poloidal current contained in a $\psi$-surface is $I = 2\pi\,C(\psi) - \oint_\ell \bm{B}\cdot d\bm{x}$, which implies that $C(\psi) = \frac{1}{2\pi} I + \text{const.}$ is the poloidal current flowing in a $\psi$-surface plus a constant. If the equilibrium $(\bm{B},\mathsf{p})$ in $Q$ is sustained by a collection of current-carrying coils outside of $Q$ then the constant is readily seen to be the total poloidal coil current divided by $2\pi$. }
\end{remark}

\section{Interpretation of the GGS equation\label{interp_sec}}
Rigidly-symmetric equilibria are completely characterized by the classical Grad-Shafranov equation. This means (a) if $(\bm{B},\mathsf{p})$ is a solution with a continuous Euclidean symmetry then it must satisfy the Grad-Shafranov equation, and (b) if $\psi$ is a solution of the Grad-Shafranov equation then it provides a solution of the equilibrium equations. In contrast, three-dimensional solutions are not completely characterized by the generalized Grad-Shafranov equation. While (a) is still true, (b) is not necessarily so. This can be seen from the derivation of the GGS equation in the proof of Theorem \ref{theorem:GS_equation}. The crucial step in the derivation was finding Eq.\,\eqref{eq:averaged_gs} by averaging force balance over the circle-action $\Phi_\theta$. While the GGS equation ensures that this averaged force balance condition is satisfied, it says nothing \emph{a priori} about the fluctuating part.

Although solutions of the Grad-Shafranov equation need not be exact solutions of the equilibrium equations, they do happen to be divergence-free. To be precise, if $\psi$ is a solution of the GGS equation defined relative to some volume-preserving circle-action $\Phi_\theta$ then $\bm{B}$ given by Eq.\,\eqref{stream_rep_final} satisfies $\nabla\cdot\bm{B} = 0$. This follows from the following simple calculation.
\begin{align}
\nabla\cdot \bm{B }  & = \nabla\cdot\left(C(\psi)\frac{\bm{u}}{R^2}\right) + \nabla\cdot\left(\frac{\rho_{\overline{g}}\overline{\nabla}\psi\overline{\times}\bm{u}}{R^2}\right)\nonumber\\
& = \rho_{\overline{g}}\overline{\nabla}\overline{\cdot}\left(\frac{\overline{\nabla}\psi\overline{\times}\bm{u}}{R^2}\right)\nonumber\\
& = -\rho_{\overline{g}}\overline{\nabla}\overline{\times}\left(\frac{\bm{u}}{R^2}\right)\overline{\cdot}\,\overline{\nabla}\psi\nonumber\\
& = 0,
\end{align}
where we have used the structure equation \eqref{force_free}, which says that $\bm{u}/R^2$ is force-free with respect to the averaged metric, on the last line. 

In light of the preceding remarks, solutions of the GGS equation represent smooth approximate solutions of the equilibrium equations. { To find such an approximate solution, it is sufficient to specify a volume-preserving circle-action $\Phi_\theta$, along with the pair of free functions $C(\psi),p(\psi)$. Then, if the GGS equation can be solved, $\bm{B}$ can be constructed using Eqs.\,\eqref{ggs_eqn} and \eqref{stream_rep_final}. }

{ \emph{A priori}, a solution of the GGS equation provides an approximate solution in a rather weak sense -- while $\nabla\cdot\bm{B} = 0$ is satisfied exactly, force balance is only satisfied on average.}
\begin{theorem}\label{approx_thm_summary}
{ Fix a domain $Q$ (diffeomorphic to the solid torus), a pair of smooth single-variable functions $C(\psi),p(\psi)$, and a volume-preserving circle action $\Phi_\theta$ with nowhere vanishing infinitesimal generator $\bm{u}$. Let $\psi:Q\rightarrow\mathbb{R}$ be an $S^1$-invariant solution of the associated generalized Grad-Shafranov equation \eqref{ggs_eqn}. The pair $(\bm{B},\mathsf{p})$, with $\mathsf{p}(\bm{x}) = p(\psi(\bm{x}))$ and $\bm{B}$ given by Eq.\,\eqref{stream_rep_final}, satisfies
\begin{gather}
\nabla\cdot\bm{B} = 0\\
(\overline{\nabla}\,\overline{\times}\,\bm{B})\,\overline{\times}\,\bm{B} = \overline{\nabla}\mathsf{p}.
\end{gather}
Equivalently, $\nabla\cdot\bm{B} = 0$ and
\begin{align}
\int_0^{2\pi} \,\nabla\Phi_\theta(\bm{x})\cdot\left[(\nabla\times\bm{B})\times\bm{B} - \nabla\mathsf{p}\right](\Phi_\theta(\bm{x}))\,d\theta = 0\label{average_force_balance}
\end{align}
for each $\bm{x}\in Q$.  
}
\end{theorem} 
\begin{remark}
{ By $S^1$-invariance of the GGS equation, if $\psi$ is a solution of the GGS equation that is not $S^1$-invariant then $\overline{\psi}(\bm{x})=\fint \psi(\Phi_\theta(\bm{x}))\,d\theta$ is an $S^1$-invariant solution.}
\end{remark}
{  However, the error associated with the approximation may also be quantified \emph{a posteriori} by measuring the size of the fluctuating part of the force-balance residual, i.e. $\bm{R} = [(\nabla\times\bm{B})\times\bm{B} - \nabla\mathsf{p}] - [(\overline{\nabla}\,\overline{\times}\,\bm{B})\,\overline{\times}\,\bm{B} - \overline{\nabla}\mathsf{p}]$. Such \emph{a posteriori} estimates may be used to formulate an optimization problem in the space of volume-preserving circle actions $\Phi_\theta$ for finding improved approximations.}

The GGS equation satisfies a variational principle. The Lagrangian density $\mathcal{L}(\bm{x},\psi,d\psi)$ is given by
\begin{align}
\mathcal{L} &= \frac{1}{2}\frac{\rho_{\overline{g}}^2\overline{\nabla}\psi\overline{\cdot}\overline{\nabla}\psi}{R^2} - \frac{1}{2}\frac{C^2(\psi)}{R^2} - p(\psi)\nonumber\\
 &+ \rho_{\overline{g}}\,D(\psi)\,(\bm{u}/R^2)\overline{\cdot}\,\overline{\nabla}\overline{\times}(\bm{u}/R^2),\label{action}
\end{align}
where $D(\psi) = \int^\psi C(\Psi)\,d\Psi$. A function $\psi$ is a solution of the GGS equation if and only if $\delta \int_Q \mathcal{L}\,d^3\bm{x} = 0$, where $\psi$ is subject to arbitrary variations that vanish on $\partial Q$. Note that the first three terms in $\mathcal{L}$ correspond to poloidal magnetic energy, minus toroidal magnetic energy, and minus internal energy, respectively. The physical interpretation of the last term is less clear. Using Proposition 11.4~of Ref.\,\onlinecite{Taylor_nl_2010}, it is straightforward to show that if (a) there are positive constants $b_p,b_c,b_d,c_p,c_c,c_d$ such that
\begin{align}
p(\psi)& \leq b_p \,|\psi| + c_p\\
C^2(\psi)& \leq b_c\,|\psi| + c_c\\
-D(\psi)\,\,(\bm{u}/R^2)\overline{\cdot}\,\overline{\nabla}\overline{\times}(\bm{u}/R^2)&\leq b_d\,|\psi| + c_d,
\end{align} 
and (b) $p,C$ are Lipshitz continuous, then the action functional $\int_Q\mathcal{L}\,d^3\bm{x}$ has a minimizer $\psi$ in the Sobolev space $H^1(Q)$ with $\psi = 0$ on $\partial Q$. Thus, for a large class of free functions $p,C$ the GGS equation can be solved.
%

Note that not all volume-preserving circle-actions are created equal when it comes to assessing accuracy. According to Theorem \ref{hidden_symmetry_thm}, for any true three-dimensional solution there is a corresponding circle-action $\Phi_\theta$ whose GGS equation kills the fluctuating part of force balance exactly. However, there is no guarantee that an arbitrary guess for $\Phi_\theta$ will have this nice property.

Because there must be \emph{some} volume-preserving circle-action that kills the fluctuating part of force balance for any exact solution, to find true equilibrium solutions it is sufficient to search through the space of volume-preserving circle-actions $\Phi_\theta$. While the space of $\Phi_\theta$'s does exhibit some topological complexity, if we restrict attention to $\Phi_\theta$ with $\bm{u}$-lines that wrap just once around the torus $Q$ (corresponding to the topological type of axisymmetry) then we have the following straightforward parameterization of the space of $\Phi_\theta$.

\begin{proposition}\label{characterization_of_circle_actions}
Suppose $\Phi_\theta$ is a volume-preserving circle-action on $Q\approx D^2\times S^1$ and that one of the $\bm{u}$-lines generates the fundamental group $\pi_1(Q)$. Then there is a diffeomorphism $\psi: D^2\times S^1\rightarrow Q:(x,y,\zeta)\mapsto \psi(x,y,\zeta)$ such that 
\begin{align}
\partial_\zeta = \psi^*\bm{u}\\
\left(\frac{\int_Q\Omega}{2\pi^2}\right)d\zeta\wedge dx\wedge dy &= \psi^*\Omega,
\end{align}
where $\Omega$ is the standard Euclidean volume form on $Q$.
\end{proposition}
\begin{proof}
Because there is a $\bm{u}$-line that generates the fundamental group, the covering map $\psi$ provided by Proposition \ref{volume_covering} in the Appendix is actually a diffeomorphism.
\end{proof}

Proposition \ref{characterization_of_circle_actions}, whose statement is rather technical, has the following intuitive interpretation. Suppose we have found a coordinate system $(x,y,\zeta)$ on $Q$ such that the Jacobian determinant is the constant $c_0 = \left(\frac{\int_Q\Omega}{2\pi^2}\right)$ and the boundary $\partial Q$ is mapped to $x^2 + y^2 = 1$. (The angle $\zeta\in\mathbb{R}\text{ mod }2\pi$.) The covariant basis vector $\partial_\zeta$ associated with this coordinate system generates a volume-preserving circle-action according to $\Phi_\theta = \exp(\theta\,\partial_\zeta):(x,y,\zeta)\mapsto (x,y,\zeta + \theta)$. Proposition \ref{characterization_of_circle_actions} says that all volume-preserving circle-actions (with $\bm{u}$-lines that have the correct topology) may be constructed in this manner. Thus, searching the space of $\Phi_\theta$ may be accomplished by searching through the space of coordinate systems $(x,y,\zeta)$ with constant Jacobian. 

We may always use the standard $(R,\phi,Z)$ cylindrical coordinates to construct one such $(x,y,\zeta)$ coordinate system. The Euclidean volume element in $(R,\phi,Z)$ coordinates is $R \,dR\,d\phi\,dZ$. Therefore if we define $(x,y,\zeta)$ according to 
\begin{align}
x &= \frac{1}{2}R^2,\quad y = Z,\quad \zeta = \phi
\end{align}
the Euclidean volume element becomes $dx\,d\zeta\,dy$. Starting from this basic set of unit-Jacobian toroidal coordinates, all other such coordinate system may be generated using the following result.

\begin{proposition}
Let $c_0 = \int_Q\Omega/(2\pi^2)$, where $Q\approx D^2\times S^1$. If $\psi_0: D^2\times S^1\rightarrow Q$ is a given diffeomorphism such that $\psi_0^*\Omega = c_0\,d\zeta\wedge dx\,\wedge dy$, then any diffeomorphism $\psi:D^2\times S^1\rightarrow Q$ with $\psi^*\Omega = c_0\,d\zeta\wedge dx\,\wedge dy$ admits the decomposition
\begin{align}
\psi = \psi_0\circ E,
\end{align}
where $E:D^2\times S^1\rightarrow D^2\times S^1$ is a volume-preserving mapping of $D^2\times S^1$.
\end{proposition}

\noindent In other words, we may search through the space of $\Phi_\theta$ by searching through the space of coordinate transformations that preserve the standard volume element $dx\,dy\,d\zeta$. Such coordinate transformations are given by maps $(x,y,\zeta)\mapsto (\overline{x},\overline{y},\overline{\zeta})$ such that $|\partial(\overline{x},\overline{y},\overline{\zeta})/\partial(x,y,\zeta)| = 1$. The volume-preserving circle action defined by such a coordinate transformation is $\Phi_\theta = \exp(\theta\,\partial_{\overline{\zeta}})$. { Recent work in the area of physics-informed neural networks\cite{JZKT_arxiv_2020,BTM_arxiv_2020} has demonstrated the feasibility of optimization in the space of symplectic coordinate transformations. Since the symplectic property is more stringent than the volume-preserving property, this suggests that one practical approach to optimizing in the space of unit-Jacobian coordinate transformations is to develop a neural network architecture whose input-to-output mapping is area-preserving.}

\section{Comparison with previous work}
The GGS equation is equivalent to force balance averaged over the angular parameter $\theta$ associated with a hidden family of volume-preserving symmetries. There is another well-known averaged force-balance equation, given for instance in Eq.\,(2.24) of Ref.\,\onlinecite{White_tc_2014}. The distinct roles of these averaged equations may be understood in Hamada coordinates $(V,\theta_1,\theta_2)$. Without loss of generality, suppose we had chosen $\bm{u} = \partial_{\theta_1}$ in Theorem \ref{hidden_symmetry_thm}. The averaging underlying the GGS equation corresponds to averaging over $\theta_1$. In contrast, the derivation of Eq.\,(2.24) from Ref.\,\onlinecite{White_tc_2014} requires averaging over both $\theta_1$ and $\theta_2$, i.e.~a flux surface average. Therefore the GGS equation implies the averaged force-balance equation, but the converse is not true; the GGS equation is a strictly stronger condition. We may also summarize this comparison in terms of Fourier harmonics $F_{nm}$ of the force-balance equation in Hamada coordinates: the GGS equation is equivalent to $F_{0m} = 0$ for all poloidal model numbers $m\in\mathbb{Z}$, while the averaged force-balance equation is equivalent to $F_{00} = 0$.

In Ref.\,\onlinecite{Boozer_pf_1981}, following Grad,\cite{Grad_conj_1967} Boozer argues that singularities generally arise in three-dimensional solutions of the MHD equilibrium equations at rational flux surfaces with rotational transform $\iota = n/m$. The singularity in the parallel current density $j_\parallel = \bm{b}\cdot\nabla\times\bm{B}$ is proportional to the product of $p^\prime$ and $\delta_{nm}$ at the rational surface, where $\delta_{nm}$ is the $(n,m)$ Fourier harmonic of $1/|\bm{B}|^2$. Averaging the force balance equation over $\Phi_\theta$ mathematically regularizes this problem because $L_{\bm{u}}\langle |\bm{B}|^2\rangle  = L_{\bm{u}} (\bm{B}\overline{\cdot}\bm{B}) = 0$, which implies the field strength defined with respect to the averaged metric only depends on a single integer combination of Hamada coordinates. This is one way to understand the lack of singularities in solutions of the GGS equation.

The classical Grad-Shafranov equation obeys a variational principle that is a special case of the variational principle for the GGS equation. (C.f. Eq.\,\eqref{action}.) We refer the reader to Ref.\,\onlinecite{Lao_1981}, where the GS variational principle is used to develop approximate solutions of the GS equation. As we mentioned in Section \ref{interp_sec}, the GGS variational principle may be used to prove the existence of solutions of the GGS equation. As such, we believe the GGS variational principle may be useful for constructing numerical approximations of such solutions. { Note that the variational principle we provide for the GGS equation is for the fixed-boundary formulation of the equilibrium problem. It would be interesting to extend our results to allow for a plasma-vacuum interface in future work.}


In this Article, we have advocated a strategy for finding exact three-dimensional solutions of the MHD equilibrium equations that involves optimization over a space of coordinate transformations with unit Jacobian determinant. In Ref.\,\onlinecite{Bhattacharjee_1984}, Bhattacharjee \emph{et. al.} propose a different optimization strategy for finding three-dimensional equilibria that also involves a search over a space of coordinate transformations. The latter reference imposes the constraint on $\bm{B}$ that it admits a system of flux coordinates. In contrast, the approach we advocate here imposes a much stronger constraint on $\bm{B}$, namely that it solves the GGS equation. (Many magnetic fields that admit flux coordinates do not satisfy the GGS equation.) That said, we have not indicated an objective functional for our optimization procedure, whereas Ref.\,\onlinecite{Bhattacharjee_1984} proposes Grad's action\cite{Grad_vp_1964} $\int_Q (|\bm{B}|^2/2 - p(\psi))\,d^3\bm{x}$. We plan to determine whether Grad's action still serves as a valid objective functional when the stronger GGS constraint is imposed in future work. An alternative objective functional would be the $L^2$-norm of the residual of force balance.

In previous work\cite{BKM_2019} we identified a quasisymmetric variant of the Grad-Shafranov equation. (See Theorem 10.5 in Ref.\,\onlinecite{BKM_2019}.) Its role in the theory of equilibria is distinct from the GGS equation discussed in this Article. Where the GGS equation applies to all nondegenerate solutions, the quasisymmetric GS equation applies specifically to equilibria that are quasisymmetric. Consequently, the GGS equation and the quasisymmetric GS equation simultaneously apply to every quasisymmetric equilibrium. As noted in Ref.\,\onlinecite{BKM_2019}, the quasisymmetric GS equation does not possess a natural variational principle unless the infinitesimal generator of quasisymmetry $\bm{u}$ satisfies $(\nabla\times\bm{u})\times\bm{u} + \nabla(\bm{u}\cdot\bm{u}) = 0$. It is therefore curious that the GGS equation always possesses a variational principle. We plan to investigate the relationship between these two equations further in future work.

An alternative notion of approximate smooth three-dimensional equilibrium solutions has been developed recently by Ginsberg, Constantin, and Drivas.\cite{Ginsberg_2020} In contrast to the approximate smooth solutions provided by the GGS equation,{ (whose approximation properties are summarized in Theorem \ref{approx_thm_summary})} these approximate solutions satisfy force balance with a small {  external force} whose size is explicitly controlled. { If Grad's conjecture is true then such approximate solutions may be the closest one can get to ``true" smooth solutions of the MHD equilibrium equations. Physically, deviations from ideal force balance are to be expected in the presence of plasma flow or effects beyond the ideal model such as pressure anisotropy. Therefore, provided the external force is comparable to the ``missing terms" in the force balance equation, these approximate smooth solutions should be considered just as adequate as an exact solution of the ideal equilibrium equations. }

{ If Grad's conjecture is true, our proposed method for finding smooth three-dimensional equilibria based on minimizing the norm of the residual $\bm{R} = [(\nabla\times\bm{B})\times\bm{B} - \nabla\mathsf{p}] - [(\overline{\nabla}\,\overline{\times}\,\bm{B})\,\overline{\times}\,\bm{B} - \overline{\nabla}\mathsf{p}]$ over the space of volume-preserving circle actions may be ill-conditioned or even ill-posed. However, the work of Ginsberg, Constantin, and Drivas (GCD) suggests that a reasonable alternative optimization objective would be to minimize the norm of $\bm{R}^\prime = [(\nabla\times\bm{B})\times\bm{B} - \nabla\mathsf{p}] - [(\overline{\nabla}\,\overline{\times}\,\bm{B})\,\overline{\times}\,\bm{B} - \overline{\nabla}\mathsf{p}] + \bm{f} $, where $\bm{f}$ is a small external force. Note that minimizing the norm of $\bm{R}^\prime$ would not result in a true solution of the equilibrium equations. Instead, minimizing the norm of $\bm{R}^\prime$ should be interpreted as a method for computing approximate equilibria similar to those constructed by GCD. We plan to investigate each of these approaches in future work. }

\section{Acknowledgement}
The authors are grateful for helpful discussions with N. Duignan, D. Ginsberg, T. Drivas, and P. Constantin. Research presented in this article was supported by the Los Alamos National Laboratory LDRD program under project number 20180756PRD4 and by the Simons Foundation (601970, RSM). Data sharing is not applicable to this article as no new data were created or analyzed in this study.

\appendix
\section{Characterization of general volume-preserving circle actions}

\begin{proposition}\label{covering_lemma}
Let $\Phi_\theta : Q\rightarrow Q$ be a free circle action on a manifold $Q\approx D^2\times S^1$ with infinitesimal generator $\bm{u}$. Assume at least one $\bm{u}$-line is not homotopic to the trivial loop. Then there exists a smooth covering map $\psi:D^2\times S^1\rightarrow Q:(x,y,\zeta)\mapsto \psi(x,y,\zeta)$ such that
\begin{align}
 \psi^*\bm{u} =\partial_\zeta.
\end{align}
\end{proposition}
\begin{proof}
First we will establish the existence of a global transverse section for the circle action, i.e.~an embedding $\eta: D^2\rightarrow Q$ such that every $U(1)$-orbit intersects $\eta(D^2)$ transversally. Choose a point $\bm{x}\in Q$ and let $\ell(\bm{x})$ be the $S^1$-orbit passing through $\bm{x}$ oriented in the obvious manner. The orbit $\ell(\bm{x})$ determines an element of the fundamental group $[\ell(\bm{x})]\in\pi_1(Q)(=\mathbb{Z})$. Because $Q$ is path connected and $\Phi_\theta$ is continuous $[\ell(\bm{x})]$ is independent of $\bm{x}\in Q$. It follows that the homology directions of $\Phi_\theta$ comprise a singleton set $\{\sigma\}$. Because at least one $S^1$ orbit is not homotopic to the trivial loop the point $\sigma$ cannot be zero. Therefore the homology directions of $\Phi_\theta$ are contained in an open half-space, and Theorem D of Ref.\,\onlinecite{Fried_1982} implies the existence of a global cross section.

Next we will use the embedding $\eta$ to explicitly construct a candidate for $\psi$ and prove the result is indeed a smooth covering map. For $(x,y)\in D^2$ and $\zeta\in S^1$ define
\begin{align}
\psi(x,y,\zeta) = \Phi_\zeta(\eta(x,y)).\label{psi_formula}
\end{align}
We will show that $\psi:D^2\times S^1\rightarrow Q$ is a surjective local diffeomorphism with the even covering property. 

\emph{Surjectivity -- } If $\bm{x}\in Q$ then because $\eta$ is a global transverse section there must be some $\zeta\in S^1$ such that $\Phi_{-\zeta}(\bm{x}) =\bm{x}^\prime \in\eta(D^2)$. Because $\bm{x}^\prime = \eta(x,y)$ for some $(x,y)\in D^2$, we therefore have $\bm{x} = \Phi_\zeta(\eta(x,y)) =\psi(x,y,\zeta)$. 

\emph{Local diffeomorphism -- } Fix $(x_0,y_0,\zeta_0)\in D^2\times S^1$. For any $\zeta\in S^1$, let $P_\zeta = \Phi_\zeta(\eta(D^2))$ be the translation of the embedded disc by $\zeta$. Note that $P_\zeta$ is a global transverse section for the circle action for each $\zeta$. The vectors $\partial_x,\partial_y\in T_{(x_0,y_0,\zeta_0)}D^2\times S^1$ are mapped to 
\begin{align}
e_x &= T_{\eta(x_0,y_0)}\Phi_{\zeta_0}T_{(x_0,y_0)}\eta[\partial_x]\\
e_y & =  T_{\eta(x_0,y_0)} \Phi_{\zeta_0}T_{(x_0,y_0)}\eta[\partial_x],
\end{align}
by the tangent mapping $T_{(x_0,y_0,\zeta_0)}\psi$.
Since $\eta$ defines a diffeomorphism $D^2\rightarrow P_0$ and $\Phi_{\zeta_0}$ restricts to a diffeomorphism $P_0\rightarrow P_{\zeta_0}$, the vectors $(e_x,e_y)$ frame the tangent space to $P_{\zeta_0}$ at $\psi(x_0,y_0,\zeta_0)$. Because the vector $\partial_\zeta$ is mapped to $e_\zeta = \bm{u}(\psi(x_0,y_0,\zeta_0))$ by $T_{(x_0,y_0,\zeta_0)}\psi$ and $\bm{u}$ is transverse to each $P_\zeta$, this shows that $(e_x,e_y,e_\zeta)$ frames the tangent space to $Q$ at $\psi(x_0,y_0,\zeta_0)$. It follows that the tangent mapping $T_{(x_0,y_0,\zeta_0)}\psi$ is invertible for any $(x_0,y_0,\zeta_0)\in D^2\times S^1$, and, by the inverse function theorem, that $\psi$ is a local diffeomorphism.

\emph{Even covering property -- } For any $\bm{p}\in \eta(D^2)$ define $\gamma_{\bm{p}}:S^1\rightarrow Q: \theta\mapsto \Phi_{\theta}(\bm{x})$. Because $\bm{u}$ is transverse to $\eta(D^2)$ the preimage $\gamma_{\bm{p}}^{-1}(\eta(D^2))$ must be a discrete subset $\{\theta_i(\bm{p})\}$ of $S^1$. Because $|\bm{u}|$ acquires a minimum value on $\eta(D^2)$, there exists a constant $r>0$ such that the balls $B_r(\theta_i)\subset S^1$ are mutually disjoint. Therefore the preimage is in fact a finite subset of $S^1$. Without loss of generality, assume that the $\theta_i(\bm{p})$ are ordered so that $\theta_i(\bm{p}) < \theta_j(\bm{p})$ whenever $i < j$. Also without loss of generality, assume $\theta_1(\bm{p}) = 0$ and that $i\in\{1,\dots,n\}$ for some positive integer $n(\bm{p})$. Because the cardinality $n(\bm{p})$ determines the homotopy type of the orbit $\gamma_{\bm{p}}$ and $\Phi_\zeta$ is continuous for each $\zeta$, $n(\bm{p}) = n$ must be independent of $\bm{p}$. Because $\gamma_{\bm{p}}$ varies smoothly with $\bm{p}$, each $\theta_i:\eta(D^2)\rightarrow S^1$ defines a smooth $S^1$-valued function $\eta(D^2)$.

For $(x_1,y_1,\zeta_1)\in D^2\times S^1$ choose an open neighborhood $U\ni (x_1,y_1,\zeta_1)$ sufficiently small to ensure $\psi|U$ is a diffeomorphism onto its image $V = \psi(U)$. Now consider the preimage $W = \psi^{-1}(V)$. Because $V$ is open by construction and $\psi$ is continuous, $W$ is an open subset of $D^2\times S^1$. In order to show that $\psi$ satisfies the even covering property we will show that $W$ is a finite disjoint union of open sets $W_i$ such that $\psi|W_i$ is a diffeomorphism onto its image for each $i$. 

To that end, first consider the preimage of $\bm{x} = \psi(x_1,y_1,\zeta_1)\in V$. Set $\bm{p} = \eta(x_1,y_1)$. We claim $\psi^{-1}(\{\bm{x}\}) $ is the disjoint union of the $n$ points $(x_i,y_i,\zeta_i)$, $i\in\{1,\dots,n\}$, where
\begin{align}
(x_i,y_i) &= \eta^{-1}(\Phi_{\theta_i(\bm{p})}(\bm{p}))\label{deck1}\\
\zeta_i &= \zeta_1 - \theta_i(\bm{p}).\label{deck2}
\end{align}
It is clear that each $(x_i,y_i,\zeta_i)$ is contained in the preimage because 
\begin{align}
\psi(x_i,y_i,\zeta_i) & = \Phi_{\zeta_i}(\eta(x_i,y_i)) = \Phi_{\zeta_1-\theta_i(\bm{p})} (\Phi_{\theta_i(\bm{p})}(\bm{p}))\nonumber\\
& = \Phi_{\zeta_1}(\eta(x_1,y_1)) = \psi(x_1,y_1,\zeta_1).
\end{align}
Conversely, suppose $(\overline{x},\overline{y},\overline{\zeta})$ is any point in $\psi_{-1}(\{\bm{x}\})$. Then, because $\Phi_{\overline{\zeta}}(\eta(\overline{x},\overline{y})) = \Phi_{\zeta_1}(\bm{p})$, we have
\begin{align}
\eta(\overline{x},\overline{y}) = (\Phi_{\zeta_1 - \overline{\zeta}}(\bm{p})),
\end{align}
which says that $\zeta_1 - \overline{\zeta} = \theta_i(\bm{p})$ for some $i\in\{1,\dots,n\}$. In other words $(\overline{x},\overline{y},\overline{\zeta})$ must be of the form given by Eqs.\,\eqref{deck1}-\eqref{deck2}.

Now define $W_i$ as the image of $U$ under the mapping $d_i:U\rightarrow D^2\times S^1:(x,y,\zeta)\mapsto (\hat{x}_i,\hat{y}_i,\hat{\zeta}_i)$, where
\begin{align}
(\hat{x}_i,\hat{y}_i) &= \eta^{-1}(\Phi_{\theta_i(\eta(x,y))}(\eta(x,y)))\\
\hat{\zeta}_i & = \zeta - \theta_i(\eta(x,y)).
\end{align}
It is simple to verify that $d_i$ is a diffeomorphism onto its image for each $i\in\{1,\dots,n\}$. The argument from the previous paragraph shows that $\cup W_i\subset W$. Therefore we will prove the even covering property as soon as we show (a) that $W\subset W_i$, and (b) that the the $W_i$ are mutually disjoint. To see that property (a) is satisfied, observe that if $(x,y,\zeta)\in W$ then there must be a point $\bm{x}\in V$ such that $(x,y,\zeta)$ is contained in $\psi^{-1}(\{\bm{x}\})$. But the previous paragraph's argument shows that $(x,y,\zeta)$ must therefore be the image of $(\psi|U)^{-1}(\bm{x})\ni U$ under one of the $d_i$. For property (b) we may merely shrink $U$ as necessary.

Now that we have shown $\psi D^2\times S^1$ is a smooth covering map we only need to show $\partial_\zeta = \psi^*\bm{u}$. But this follows from
\begin{align}
(\psi^*\bm{u})(x,y,\zeta) = (T_{(x,y,\zeta)}\psi)^{-1}[\bm{u}(\psi(x,y,\zeta))],
\end{align}
and $T_{(x,y,\zeta)}\psi[\partial_\zeta] = \bm{u}(\psi(x,y,\zeta))$.

\end{proof}

\begin{proposition}
Let $\psi,\overline{\psi}:D^2\times S^1\rightarrow D^2\times S^1$ be a pair of smooth covering maps such that $\psi_*\pi_1(D^2\times S^1) = \overline{\psi}_*\pi_1(D^2\times S^1)$. Then there is a diffeomorphism $E:D^2\times S^1\rightarrow D^2\times S^1$ such that $\psi(x,y,\zeta) = \overline{\psi}(E(x,y,\zeta))$.
\end{proposition}

\begin{proposition}\label{volume_covering}
If $\Phi_\theta$ is a volume-preserving circle-action on $Q \approx D^2\times S^1$ with at least one orbit that is not homotopic to the trivial loop then there is a covering map $\psi:D^2\times S^1\rightarrow Q:(x,y,\zeta)\mapsto \psi(x,y,\zeta)$ such that 
\begin{align}
\partial_\zeta &= \psi^*\bm{u}\label{straight_u}\\
\left(\frac{\int_Q\Omega}{2\pi^2}\right)d\zeta\wedge dx\wedge dy &= \psi^*\Omega,\label{flat_volume}
\end{align}
where $\Omega$ is the standard Euclidean volume form on $Q$.
\end{proposition}

\begin{proof}
By Lemma \ref{covering_lemma} we may find a covering map $\overline{\psi}$ such that $\overline{\psi}^*\bm{u} = \partial_{\overline{\zeta}}$. We will find a diffeomorphism $E:D^2\times S^1\rightarrow D^2\times S^1$ such that $\psi = \overline{\psi}\circ E$ satisfies Eqs.\,\eqref{straight_u}-\eqref{flat_volume}.

In order to ensure $\partial_\zeta = \psi^*\bm{u} = E^*\overline{\psi}^* \bm{u} = E^*\partial_{\overline{\zeta}}$ the diffeomorphism $E=(x,y,\zeta)$ must satisfy
\begin{align}
x(\overline{x},\overline{y},\overline{\zeta} + \theta)& =x(\overline{x},\overline{y},\overline{\zeta} ) \\
y(\overline{x},\overline{y},\overline{\zeta} + \theta) & =y(\overline{x},\overline{y},\overline{\zeta} )\\
\zeta(\overline{x},\overline{y},\overline{\zeta} + \theta) & = \zeta(\overline{x},\overline{y},\overline{\zeta} )+ \theta.
\end{align} 
These conditions will be satisfied if and only if $(x,y) = \phi(\overline{x},\overline{y})$ for some diffeomorphism $\phi:D^2\rightarrow D^2$ and $\zeta(\overline{x},\overline{y},\overline{\zeta}) = \overline{\zeta} + S(\overline{x},\overline{y})$ for some smooth function $S:D^2\rightarrow S^1$.

Because $L_{\bm{u}}\Omega = 0$, we know in advance that $L_{\partial_{\overline{\zeta}}}\overline{\psi}^*\Omega = 0.$ Writing $\overline{\psi}^*\Omega = \overline{\rho}\,d\overline{\zeta}\wedge d\overline{x}\wedge d\overline{y}$, this means $\partial_{\overline{\zeta}}\overline{\rho} = 0$ or $\overline{\rho} = \overline{\rho}(\overline{x},\overline{y})$. Note that $\overline{\rho}$ also satisfies $2\pi \int_{D^2}\overline{\rho}\,d\overline{x}\,d\overline{y} = \int_Q\Omega$. Therefore, assuming $E$ is of the form determined in the previous paragraph, the pullback $\psi^*\Omega$ is given by
\begin{align}
\psi^*\Omega &= E^*\overline{\psi}^*\Omega = E^*(\overline{\rho}\,d\overline{\zeta}\wedge d\overline{x}\wedge d\overline{y}) \nonumber\\
&= d\zeta\wedge \phi^*(\overline{\rho}\,d\overline{x}\wedge d\overline{y}).
\end{align}
Note that, regardless of the form of $\phi$, we must have
\begin{align}
\int_{D^2}\phi^*(\overline{\rho}\,d\overline{x}\,d\overline{y}) = \frac{1}{2\pi}\int_Q\Omega.
\end{align}
According to Moser's theorem,\cite{Moser_1965} if $\lambda dx\wedge dy$ is any volume form on $D^2$ with the same total volume as $\overline{\rho}\,d\overline{x}\wedge d\overline{y}$ then there is a diffeomorphism $\phi:D^2\rightarrow D^2$ such that $\phi^*(\overline{\rho}\,d\overline{x}\wedge d\overline{y}) = \lambda dx\wedge dy$. In particular, we may find such a $\phi$ for $\lambda = \int_Q\Omega / (2\pi^2)$, which gives the desired result. Note that $E$ is not unique; $S$ is completely free, and $\phi$ is only fixed modulo volume-preserving conjugations of $D^2$.
\end{proof}



\providecommand{\noopsort}[1]{}\providecommand{\singleletter}[1]{#1}%
%


\end{document}